\documentclass[10pt,twocolumn]{article}
\usepackage{geometry}                
\usepackage{graphicx}
\usepackage{amssymb}
\usepackage{epstopdf}
\usepackage{amssymb}
\usepackage{amsmath, amsthm}
\usepackage{color}
\usepackage{epstopdf}
\usepackage{epsfig}

\newtheorem{theorem}{Theorem}

\newtheorem{lemma}[theorem]{Lemma}

\newtheorem{observation}[theorem]{Observation}
\newtheorem{defn}[theorem]{Definition}

\title{Walks on SPR Neighborhoods}
\author{Alan Joseph J.~Caceres\thanks{Dept. of Math \& Computer Science, Lehman College, City University of New York,
Bronx, New York, 10468; Corresponding author:  {\tt stjohn@lehman.cuny.edu};
Support provided by NSF Grant \#09-20920.}
\thanks{Supported by an
undergraduate research fellowship from the Louis Stokes Alliance for Minority Participation in Research, NSF \#07-03449.} 
\and Juan Castillo\footnotemark[1] \footnotemark[2] 
\and Jinnie Lee\footnotemark[1] 
\and Katherine St.~John\footnotemark[1]}

\begin{document}
\maketitle

\begin{abstract}
A nearest-neighbor-interchange (NNI) walk is a
sequence of unrooted phylogenetic trees, $T_0, T_1, T_2, \ldots$
where each consecutive pair of trees differ by a single
NNI move.  We give tight bounds on the length of the shortest 
NNI-walks that
visit all trees in an subtree-prune-and-regraft (SPR) neighborhood
of a given tree.  
For any unrooted, binary tree, $T$, on $n$ leaves,
the shortest walk takes $\theta(n^2)$ additional
steps than the number of trees in the SPR neighborhood.
This answers Bryant's Second Combinatorial Conjecture 
from the Phylogenetics Challenges List, the Isaac Newton Institute, 2011, 
and the Penny Ante Problem List, 2009.
\end{abstract}

\section{Introduction}

\begin{figure*}
\begin{center}
    \begin{tabular}{ccccc}  
            \includegraphics[width=.725in]{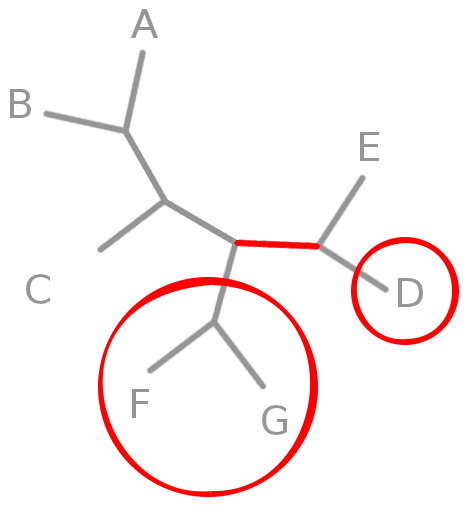}
            \begin{tabular}{c}
                $\leftarrow$\\
                \\
                \\
                \\    
            \end{tabular}
        &
        \includegraphics[width=.725in]{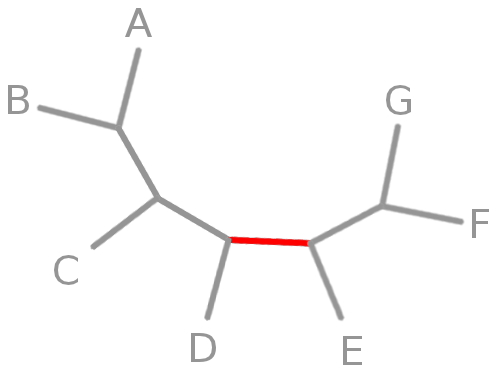} \begin{tabular}{c}
                $\leftarrow$\\
                \\
                \\
                \\    
            \end{tabular}&
        \includegraphics[width=.725in]{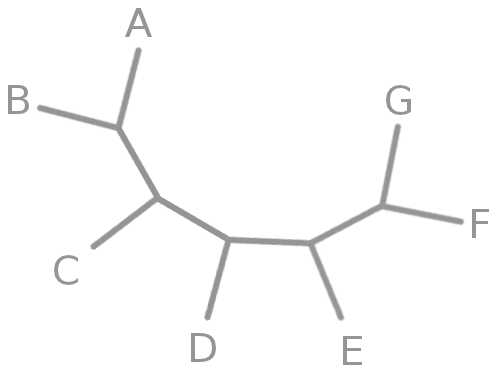} &
        \begin{tabular}{c}
                $\rightarrow$\\
                \\
                \\
                \\    
            \end{tabular} \includegraphics[width=.725in]{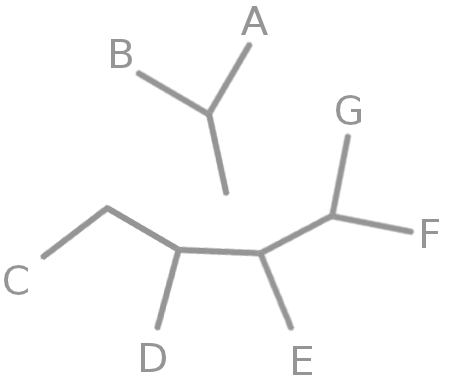} &
        \begin{tabular}{c}
                $\rightarrow$\\
                \\
                \\
                \\    
            \end{tabular} \includegraphics[width=.725in]{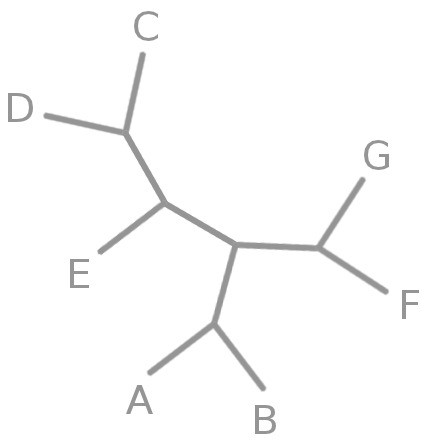} \\
    \end{tabular}
\end{center}
\caption{\small The trees on the left and center differ by a single NNI move.
The tree on the right differs by a single SPR move from the center tree.}
\label{spr_fig}
\label{nniFigure}
\end{figure*}

Evolutionary histories, or phylogenies, are essential 
structures for modern biology \cite{hillisBook}.  Finding the
optimal phylogeny is NP-hard, even when we restrict to 
tree-like evolution \cite{FouldsGraham,rochMLNP}.  As such,
heuristic searches are used to search the vast set of all trees.
There are many search techniques used (see \cite{whelan2007} for 
a survey), but most rely on local search.  That is, at each
step in the search, the next tree is chosen from the ``neighbors''
of the current tree.  A popular way to define neighbors is 
in terms of the subtree-prune-and-regraft (SPR) metric 
(defined in Section~\ref{def-section}). 
For a given unrooted tree on $n$ leaves,
or taxa, the SPR-neighborhood is the number of trees that are differ
by a single SPR move. 
The number of trees in the SPR neighborhood is $(2n-6)(2n-7)$.
The second ``Walks on Trees'' conjecture of Bryant \cite{bryantConjecture,INIChallenges} focuses on 
efficiently traversing this neighborhood via the nearest-neighbor-interchange
(NNI) tranformations (defined in Section~\ref{def-section}).  Bryant asks:
\begin{quote}
    An  “NNI-­‐walk”  is  a  sequence  $T_1,  T_2, \ldots,  T_k$  of
 unrooted  binary  phylogenetic  trees   where  each  consecutive
 pair  of  trees  differ  by  a  single  NNI. 

i. [Question] What  is  the  shortest  NNI  walk  that  passes  through  all  
binary  trees  on  $n$  leaves?

ii. [Question]  Suppose  we  are
 given  a  tree  $T$.  What  is  the  shortest  NNI   walk  that
 passes  through  all  the  trees  that  lie  at  most  one  SPR
 (subtree   prune  and  regraft)  move  from  $T$?  
\end{quote}

Bryant's conjectures were posed as part of the New Zealand Phylogenetic
Meetings' Penny Ante Problems \cite{bryantConjecture} as well as the 
Challenges problems from the most recent Phylogenetics Meeting
at the Isaac Newton Institute \cite{INIChallenges}.

We answer the second question, proving that the shortest walk takes
$\Theta(n^2)$ more steps than the theoretical minimum that visits
every tree exactly once (that is, a Hamiltonian path).  This builds on
past work \cite{walksPaper} that showed that a Hamiltonian path was
not possible.

\section{Background}
\label{def-section}

This section includes definitions and results that we use 
from Allen \& Steel \cite{allenSteel}.  For a more detailed background
on mathematically phylogenetics, see
Semple \& Steel \cite{sempleSteel}.

\begin{defn}
An unrooted binary phylogenetic tree (or more briefly a binary tree) is a tree whose leaves (degree 1 vertices) are labelled bijectively by a (species) set $S$, and such that each non-leaf vertex is unlabelled and has degree three. We let $UB(n)$ denote the set of such trees for $S = \{1,\ldots,n\}$.
\end{defn}

Each internal edge, $e$ of a tree $T \in UB(n)$ yields a natural bipartition, or {\bf split} of the taxa;  We write $A\mid B$ if there is an 
edge which partitions the leaf set into the two sets $A$ and $B$.
We use the standard notation of $T_A$ to refer to the smallest subtree of $T$ containing
leaves only from $A$.

Figure~\ref{spr_fig} shows several binary trees. 
Each edge of a tree induces a {\bf split} of the leaf set $S$.
The \textit{Nearest Neighbor Interchange} (NNI) is a distance metric introduced independently by DasGupta {\em et al.} \cite{dasgupta2000} and Li {\em et al.} \cite{li1996}. 
Roughly, an NNI operation swaps two subtrees that are separated by an internal edge in order to generate a new tree \cite{allenSteel}.

\begin{defn} Allen and Steel \cite{allenSteel}:
Any internal edge of an unrooted binary tree has four subtrees attached to it. A {\bf nearest neighbor interchange} (NNI) occurs when one subtree on one side of an internal edge is swapped with a subtree on the other side of the edge, as illustrated in Figure~\ref{nniFigure}.
\end{defn}

\begin{defn}
The {\bf NNI distance}, $d_{NNI}(T_1,T_2)$, between two trees $T_1$ and $T_2$ is defined as the minimum number of NNI operations required to change one tree into the other.
\end{defn}

The complexity of computing the NNI distance was open for over 25 years, and was proven to be NP-complete by Allen and Steel \cite{allenSteel}. For a tree with $n$ uniquely labeled leaves, there are $n-3$ internal branches. Thus, there are $2(n-3)$ NNI rearrangements for any tree.

One of the most popular moves used to search treespace is the Subtree-Prune-and-Regraft (SPR).  Roughly, an SPR move prunes  a  selected subtree and then reattaches it  on an edge selected from the remaining tree (see Figure~\ref{spr_fig}).

\begin{defn} Allen and Steel \cite{allenSteel}:
A {\bf subtree prune and regraft} (SPR) on a phylogenetic tree $T$ is defined as cutting any edge and thereby pruning a subtree, $t$, and then regrafting the subtree by the same cut edge to a new vertex obtained by subdividing a pre-existing edge in $T$--$t$. We also apply a forced contraction to maintain the binary property of the resulting tree  (see Figure~\ref{spr_fig}).
\end{defn}

\begin{defn}
The {\bf SPR distance}, $d_{SPR}(T_1,T_2)$, between two trees is the minimal number of SPR moves needed to transform the first tree into the second tree. 
\end{defn}

\begin{figure*}
\label{nbhd-fig}
\begin{center}
\begin{tabular}{cc}
\includegraphics[height=1.5in]{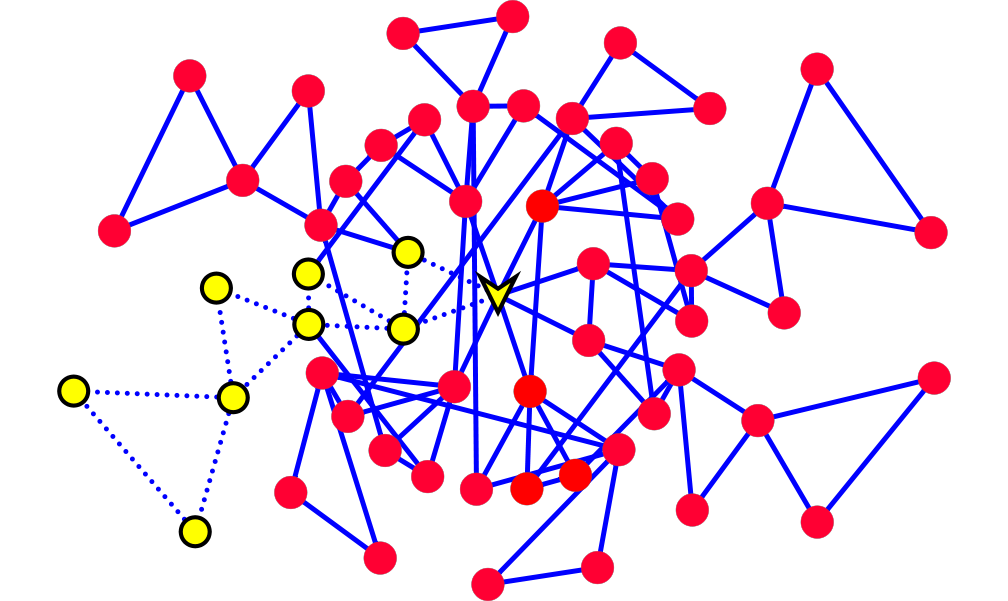} \hfill
\includegraphics[height=1.5in]{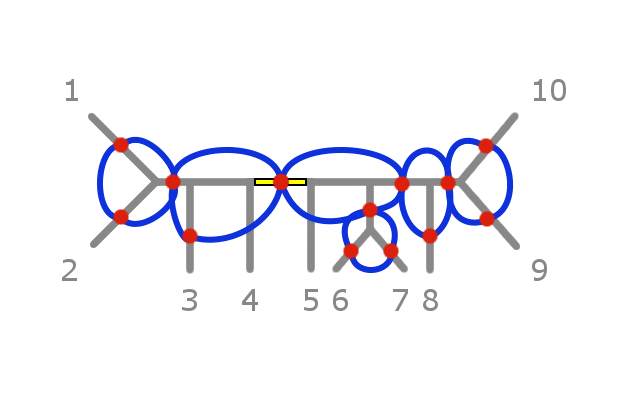}
\end{tabular}
\end{center}
\caption{\small Left: The SPR-neighbor of a 7-leaf caterpillar tree. The highlighted nodes show
the trees in the orbit that prunes a leaf from one of the sibling pairs.  We note that any $NNI$-walk that visits 
every tree in this SPR neighborhood will visit some trees more than once.
Right: The orbit of the edge,
$e = 1,2,3,4 \mid 5,6,7,8,9,10$
from the tree $(1,(2,(3,(4,(5,((6,7),(8,(9,10))))))))$ shown with respect to the tree.
The tree is shown in the background with edge $e$ highlighted.  An SPR move is determined
both by the edge pruned and the target edge of the regrafting.  The trees in the orbit
(red dots) are shown relative the regrafting edge in the initial tree.  The blue lines
indicate NNI-edges in the orbit.  We note that the edges adjacent to $e$ yield the initial
tree when used as the target edge, so, do not produce a new tree.}
\end{figure*}

The calculation of SPR distances has been proven NP-complete for both rooted and unrooted trees \cite{bordewichSemple,hickey2008}. Approximation algorithms for SPR on rooted trees exist \cite{approx,bordewich2008}.

\begin{defn}   Let $T_0$ be an unrooted, binary tree.
Define $N_{SPR}(T_0)$ to be the {\bf SPR-neighborhood}
of $T_0$, namely,
$$
  N_{SPR}(T_0) = \{ T \mid d_{SPR}(T_0, T) = 1 \}
$$
\end{defn}

When the tree in question is obvious, we will drop the argument and call
the neighborhood $N_{SPR}$.

\begin{defn}   Let $T_0$ be an unrooted, binary tree and 
and $S$ be a set of trees that are 1 SPR move from $T_0$.
Define $N_{NNI}(S, T_0)$ to be the {\bf NNI-neighbors}
of $S$, namely, 
$$
  \begin{array}{rl}
  N_{NNI}(S, T_0) = &\{ T \mid \exists T' \in S, d_{NNI}(T, T') = 1 \\
                    &\mbox{ and } d_{SPR}(T_0, T') = 1 \}
  \end{array}
$$
\end{defn}

We note that for every subset $S$ of the SPR neighborhood of $T_0$,
$S \subseteq N(S)$.  

A ``sibling pair'' or ``cherry'' in a tree are two leaves that have 
the same parent.
A ``caterpillar tree'' refers to the unrooted tree with exactly 2 sibling pairs.

\section{Results}

Theoretically, the shortest walk of the SPR neighborhood would be if each tree could be
visited exactly once (that is, a ``Hamiltonian path'').  In \cite{walksPaper}, 
this was shown to be impossible for $n \geq 6$.  
This was done by 
showing that in the SPR neighborhood of caterpillar trees, there are 
at least 4 `isolated triangles' on the outer edge of the neighborhood 
(see Figure~\ref{nbhd-fig}) that force at least
two trees to be visited twice.

To bound the number of steps needed to visit the SPR neighborhood, we 
first introduce a new concept of an orbit of an edge.  The orbit is all trees
created by breaking that edge (in either direction) in an SPR move.  More
formally:

\begin{defn}  Define for each edge $e$ of the tree $T_0$, the {\bf orbit} of $e$, 
$O_e$,
to be all the trees that are one SPR move from $T_0$ where the edge broken
by the SPR move is $e$.  
\end{defn}

In Figure~\ref{nbhd-fig}, the SPR-neighbohood, as well as orbits, of
the 7-taxa caterpillar tree are shown.  

To calculate the size of the SPR neighborhood of a tree, 
Allen and Steel (proof of Theorem~2.1, \cite{allenSteel}) characterized 
the relationship between the trees in the neighborhood.  

\begin{theorem} Allen and Steel \cite{allenSteel}: Let $T_0$ be an unrooted
  phylogenetic tree on $n$ leaves and let $N_{SPR} = N_{SPR}(T_0)$ be all trees
  that are a single SPR move from $T_0$.
  \begin{enumerate}
    \item The size of the SPR neighborhood is $|N_{SPR}| = 2(n-3)(2n-7)$. 
    \item The number of trees in $N_{SPR}$ that can be 
        obtained by more than one SPR move from $T_0$ are exactly 
        those from the NNI transformations. Thus, there are $2n-6$ of them.
    \item The number of trees in $N_{SPR}$ that can be
        obtained by only one SPR move from $T_0$ is $4(n-3)(n-4)$.
  \end{enumerate}
   \label{as-thm}
\end{theorem}

From this, we make the following simple observations about orbits:

\begin{observation}  Let $T_0$ be an unrooted phylogenetic tree on $n$ leaves:
  \begin{enumerate}
    \item Every tree $T \in N_{SPR}$ belongs to some
      orbit $O_e$ for $e \in E(T_0)$.
    \item Each orbit contains $T_0$.
    \item Excluding $T_0$, there are exactly $2n-6$ trees that are included
        in two orbits.
    \item The number of orbits is the number of edges in the tree, $2n-3$.
    \item The size of each orbit is $2n-5$.
  \end{enumerate}
  \label{obs}
\end{observation}

The SPR neighborhood is the union of the orbits, but surprisingly, these
orbits are mostly disjoint. Roughly,
the overlap of orbits is very small and they have very few neighbors in
common.  Formally:
  
\begin{lemma}  Let $T_0$ be an unrooted phylogenetic tree on $n$ taxa.
    Let $T_1, T_2 \in N_{SPR}(T_0)$ and there exists $e \in E(T_0), T_1,T_2 \in O_e$.  
    Let $e_i$ be the target edge of
    the move that created $T_i$ for $i=1,2$
    (that is, $T_1$ is formed by grafting some
    pruned subtree of $T_0$ to $e_1$ and $T_2$ is the result of grafing a 
    pruned subtree to $e_2$).

    Then, $T_1$ and $T_2$ differ by a single NNI move if and only if 
    $e_1$ and $e_2$ have a common endpoint in $T_0$.
    \label{overlap-lemma}
\end{lemma}

\begin{proof}

    $\Longleftarrow$:  Assume that $e_1$ and $e_2$ have a common endpoint in $T_0$.
        Let $M$ be the subtree pruned by the SPR move that creates $T_1$.  Without
        loss of generality, let the split induced by $e_1$ in $T_0$ be $A B C \mid D E M$
        and the split induced by $e_2$ in $T_0$ be $A B \mid C D E M$, where 
        $A$,$B$,$C$,$D$,$E$, and $M$ are the leaves of subtrees of $T_0$.  Let
        $T_X$ refer to the subtree with leaves only from the set $X$.

        If $T_M$ is pruned to create $T_1$, then we have that $T_1$ contains the 
        splits: $A B M \mid C D E$ and $A B \mid M C D E$.  
        If $T_M$ is also pruned to create $T_2$, then we have that $T_2$ contains the
        splits: $A B C M \mid D E$ and $A B C\mid M D E$. Thus, $T_1$ and $T_2$ differ
        by a single NNI move (swapping $T_C$ and $T_M$), and the hypothesis holds.

        So, assume that $T_M$ is not pruned to create $T_2$, but instead that $e$ is 
        pruned in the other direction.  Let $N = S - M$, where $S$ is the set of leaves
        of $T$.  
        Since $e_1$ and $e_2$ 
        share an endpoint, at least one of them must be the edge pruned, $e$.  If both
        are $e$, then $T_1 = T_2 = T$, and the hypothesis is trivially true.  
        If only one, say $e_2$, is $e$, then $e_1$ must be a neighbor of $e$ in $T$
        which implies $T_2 = T_1 =T$, and again the hypothesis is trivially true.

    $\Longrightarrow$:  
        By assumption $T_1$ and $T_2$ differ by a single NNI move.  By definition of 
        the NNI move, there exists an edge $e'\in E(T_1)$ that
        when removed, breaks $T_1$ into 4 distinct subtrees, $T_A,T_B,T_C,T_D$ with
        leaf sets, $A,B,C,D$,
        and the split $A B \mid C D$ belongs to $T_1$ while $B C \mid A D$
        belongs to $T_2$.  Since both $T_1$ and $T_2$ are in the same
        orbit, the same edge $e$ is pruned to create both.  We note that 
        since they differ by only the NNI move, that, by the argument above,
        the pruning of $e$ must occur in the same direction for both to be result in 
        non-trivial trees.  Further, $e$ must prune one of the subtrees: $A,B,C,D$, 
        since only one move is allowed and $T_1$ and $T_2$ contain exactly the same trees.  
        Without loss of generality, assume that $A$ is pruned.  We note
        the trees induced by the leaves of $B,C,D$ are identical for
        these trees:  $T_0 |_{L(B\cup C\cup D)} = T_1 |_{L(B\cup C\cup D)} = T_2 |_{L(B\cup C\cup D)}$.
        It follows that $e_1$ and $e_2$ share a common endpoint, namely 
        the intersection point of $B,C,D$.
\end{proof}

We can immediately give an upper bound on the NNI-walk of the SPR 
neighborhood as $O(n^2)$ steps.  The underlying idea is to traverse
each orbit separately, and then link these paths to form a traversal
of the entire SPR neighborhood:

\begin{lemma}  The SPR neighborhood has an NNI-walk of length 
    $O(n^2)$.
    \label{upper-bd-lemma}
\end{lemma}

\begin{proof} We will break the NNI-walk of the SPR neighborhood
into a NNI-walk of the orbit of each edge in $T_0$.  Since each
orbit contains the initial tree $T_0$, we can glue together the
walks of the orbit to make a walk of the entire space.  We note
that since each orbit contains at most $2n-5$, walking the $2n-3$
orbits in this fashion yields a walk with the number of
steps is bounded by $2(2n-5)(2n-3) = O(n^2)$.

It suffices to show that there is a $2$-walk of each orbit $O_e$
for $e \in E(T_0)$.  Each tree, $T\in O_e$, is created by pruning
the edge $e$ in $T_0$ and regrafting the pruned subtree to another
edge in $T_0$ (see Figure~\ref{nbhd-fig}).  Every tree in the orbit
corresponds to an edge in $T_0$ (namely, the target edge), and the
trees in the orbit are connected exactly when their target edges
share an endpoint in $T_0$ by Lemma~\ref{overlap-lemma}.  Thus,
the orbit can be traversed by at most $2(2n-3)$ steps by starting
at $T_0$ and following a depth-first-search of the tree (each tree
in the orbit is visited at most once on the way ``down'' the search
and once on the way ``up'' the search).
\end{proof}

To show the lower bound takes more work.  
It follows from this lemma that every orbit has very small overlap with 
the other orbits:

\begin{lemma}  Let $T_1, T_2 \in N_{SPR}(T_0)$ such that $T_1$ and $T_2$
    are a single NNI move apart.  Then 
    $d_{NNI}(T_0,T_1), d_{NNI}(T_0,T_2) \leq 2$.
    \label{isolated-lemma}
\end{lemma}

\begin{proof}
We note that if there exists an $e\in E(T_0)$ such that $T_1,T_2\in O_e$,
then by Lemma~\ref{overlap-lemma}, the lemma holds.

So, let us assume that there exists $e_1,e_2 \in E(T_0)$,
$T_1 \in O_{e_1}$, $T_1 \not \in O_{e_2}$, $T_1 \in O_{e_1}$ and $T_1 \not \in O_{e_2}$.
Let $M_i$ be the leaves of the subtree pruned with $e_i$ from $T_0$ to create tree $T_i$, $i=1,2$.
Since $T_1$ and $T_2$ are a single NNI move apart.  By definition, there exists a split in 
$T_1$, $A B \mid C D$ that is rearranged in $T_2$: $B C \mid A D$.

We will argue, by cases, that both 
$T_1$ and $T_2$ are within 2 NNI moves of $T_0$.
Without loss of generality, we will assume that $M_1 \cap T_A \neq \emptyset$

{\bf Case 1:}  $M_1 \subsetneq T_A$.
    Then, let $A' = A - M_1$.  So, we have that $T_1$ contains the split
    $A' M_1 B | C D$ and $T_2$ contains the split $B C \mid A' M_1 D$.
    Since $T_1$ is only one SPR move from $T_0$, the structure of the 2 trees
    is identical without $M_1$, that is, 
    $T_1 |_{A'\cup B \cup C \cup D} = T_0 |_{A' \cup B \cup C \cup D)}$,
    and $T_0$ includes and edge that splits $A'$ and $B$ from $C$ and $D$.
    Since $T_2$ does not contain such an edge, the move that creates it
    must prune one of $T_{M_1}$, $T_{A'}$, $T_B$, $T_C$, or $T_D$.  Pruning $T_{M_1}$ is
    not possible since $T_1$ and $T_2$ are in different orbits.  Pruning
    $T_{A'}$ is only possible if $T_{M_1}$ and $T_B$ have the same parent, in which
    case the $d_{NNI}(T_0,T_1), d_{NNI}(T_0,T_2) = 2$ and the lemma holds.
    Pruning $T_B$ to create $T_2$ means that the subtree $T_{M_1}$ is on the
    edge separating $A'$ and $B$ from $C$ and $D$ in $T_0$ and 
    $d_{NNI}(T_0,T_1), d_{NNI}(T_0,T_2) = 2$.  Lastly, pruning $T_B$, $T_C$, or
    $T_D$ is only possible if $T_0 = T_1$, in which case, $d_{NNI}(T_0,T_1) = 0, d_{NNI}(T_0,T_2) = 1$.
    
{\bf Case 2:} $M_1 = A$.
    So, we have that $T_1$ contains the split
    $M_1 B | C D$ and $T_2$ contains the split $B C \mid M_1 D$.
    We have three possibilities for $T_0$, namely, it could contain
    one of the following three splits:  $M_1 B \mid C D$, $B C\mid M_1 D$,
    or $B D\mid M_1 C$.  In each of these cases, we have 
    $d_{NNI}(T_0,T_1), d_{NNI}(T_0,T_2) \leq 1$ and the lemma holds.

{\bf Case 3:} $M_1 \supsetneq A$.
    So, $M_1 \cap B \neq \emptyset$.  Since $T_{M_1}$ is a subtree of $T_1$ and
    of $T_2$, it must contain all of $B$.  If $M_1 = A \cup B$, then the target edges
    in $T_1$ and $T_2$ must separate $C$ and $D$, and are identical. Similarly, 
    if $M_1 \subsetneq A \cup B$, $M_1$ must contain all of $C$ or all of $D$,
    and the taget edges in $T_1$ and $T_2$ must preserve the rooting of the 
    remaining subtree, and thus, are identical.  Thus,
    $d_{NNI}(T_0,T_1), d_{NNI}(T_0,T_2) = 0$.

\end{proof}

\begin{lemma}  Let $U \subseteq O_e$ be connected
    consist of trees more than $2$ NNI moves from $T_0$.
    Let $n = |U|$.
    Then any NNI-circuit of $U$ 
    takes at least ${\frac{3}{2}(n-1)}$ steps.
    \label{circuit-lemma}
\end{lemma}

\begin{proof}
    By induction on the size of $|U|$.

    For $|U|=1$:  This is trivially true.

    For $|U|>1$, choose $x \in U$ closest to $T_0$.  Since $x$ is closest to $T_0$,
        not all of the neighbors of $x$ are in $U$ (if so, then there is an element
        in $U$ closer to $T_0$).  Since $T_0$ is binary, it has at most 4 neighbors 
        in $N_{SPR}$.  If $x$ has one neighbor in $U$, then, a circuit of $U$ must 
        traverse the same edge from $x$ to its neighbor twice, and the number of
        steps needed is at least two more than
        the number of steps needed for the smaller set $|U| - \{x\}$.
        By inductive hypothesis, this smaller set takes at least ${\frac{3}{2}(|U - \{x\}| -1)}$ steps.
        So, the number of steps for $U$ is:
   
        $$
                \frac{3}{2}(|U-\{x\}|-1)+2
                \geq \frac{3}{2}(|U|-1)
        $$

        If $x$ has two neighbors in $U$, then call the subtrees rooted
        at $x$'s neighbors, $U_1$ and $U_2$.  If the neighbors of $x$ are
        connected, then it takes 3 steps to visit $x$ in a circuit of $x$, $U_1$, 
        and $U_2$.  If they are not connected, 
        it takes 4 steps.  Thus, by inductive hypothesis, the number of steps needed is:
            $$\frac{3}{2}(|U_1|-1)+\frac{3}{2}(|U_2|-1)+3
                \geq \frac{3}{2}(|U|-1) $$

        If $x$ has 3 neighbors in $U$, then by similar argument, we have the lower
        bound.

        If $x$ has 4 neighbors in $U$, then it is not the closest element of $U$ to
        $T_0$, giving a contradiction.
\end{proof}

From the last two lemmas, we have that the orbits are mostly isolated-- the only
trees having neighbors from outside the orbit are within 2 steps of $T_0$.  Each 
of these isolated arms of the orbit must be visited in an NNI walk of the SPR 
neighborhood, and the walks of the isolated arms take many extra steps.  This yields our 
lower bound:

\begin{lemma}  It takes $\Omega(n)$ extra steps to make a circuit of an orbit.
    \label{lower-bd-lemma}
\end{lemma}

\begin{proof}
    Let $e\in E(T_0)$ and $O_e$ its orbit.
    Since each orbit has $2n-5$ trees (Observation~\ref{obs}) and by
    Lemma~\ref{overlap-lemma}, at most $8$ have neighbors from $N_{SPR} - O_e$.

    It follows from  Lemma~\ref{overlap-lemma}, the $2n-13$ remaining trees are in two connected
    sets.  By the Pigeonhole Principal, one set has at least $n-7$ trees.
    By Lemma~\ref{circuit-lemma}, it takes $\Omega(((n-7)-1)/2) = \Omega(n)$ extra steps
    to visit the larger connected set.

    Thus, it takes $\Omega(n)$ extra steps to traverse the orbit.
\end{proof}

The above lemmas can be combined to show that $\theta(n^2)$ extra
steps are needed to traverse the neighborhood, since there are $2n-3$
orbits, and each has minimal overlap with other orbits.  

\begin{theorem}  Every SPR neighborhood takes $(2n-6)(2n-7) + \Theta(n^2)$
steps to traverse.
\end{theorem}

\begin{proof}
    The upper bound follows by Lemma~\ref{upper-bd-lemma}.
    
    For the lower bound:
    by Lemma~\ref{isolated-lemma}, every orbit, $O_e$ has $\Omega(n)$ trees
    that have no neighbors in other orbits.   
    By Lemma~\ref{lower-bd-lemma}, it takes $\Omega(n)$ extra steps to
    traverse these regions of $O_e$.
    Since, by Theorem~\ref{as-thm}, there are
    $2n-3$ orbits, we have that any path must take $\geq (2n-3)\Omega(n) = \Omega(n^2)$
    extra steps.
\end{proof}

\section{Acknowledgments}

For spirited discussions, we would like to thank the members of the 
Treespace Working Group at CUNY:
Ann Marie Alcocer,
Kadian Brown,
Samantha Daley,
John DeJesus,
Eric Ford, 
Michael Hintze,
Daniele Ippolito,
Joan Marc,
Oliver Mendez,
Diquan Moore, 
and 
Rachel Spratt.

This work was supported by grants from the US National Science 
Foundation programs in mathematical biology and computational mathematics
(NSF \#09-20920) and the New York City Louis Stokes Alliance for Minority Participation in 
Research (NSF \#07-03449).
Some of this work was completed at the Isaac Newton Institute for 
Mathematical Sciences at Cambridege University, and we gratefully 
acknowledge travel funding
to the last author provided by the institute and the
Lehman College-IBM Research travel fund.

\small
\bibliographystyle{plain}
\bibliography{treespace}

\end{document}